\documentclass[12pt,onecolumn]{IEEEtran}
\usepackage{subfigure, epsfig, color}
\usepackage{amsmath}
\usepackage{amsthm}
\newtheorem{theorem}{Theorem}[section]
\newtheorem{lemma}[theorem]{Lemma}

\begin{document}

\title{Optimizing the Number of Fog Nodes for $\textit{Cloud-Fog-Thing}$ Networks}
\author{
  \IEEEauthorblockN{
    Eren Balevi, and
    Richard D. Gitlin\\} 
    \IEEEauthorblockA{Department of Electrical Engineering, University of South Florida \\
		Tampa, Florida 33620, USA\\
      erenbalevi@mail.usf.edu, richgitlin@usf.edu }
}
\maketitle 
\normalsize

\begin{abstract}
Going from theory to practice in fog networking raises the question of the optimum number of fog nodes that will be upgraded from the existing nodes. This paper finds the optimum number of fog nodes for a given total number of ordinary nodes residing in the area of interest for different channel conditions. Determining the optimum number of fog nodes is quite beneficial, because it can strongly affect the SINR, and thus the average data rate and transmission delay. The numerical results indicate that the average data rate increases nearly an order of magnitude for an optimized number of fog nodes in case of shadowing and fading. It is further shown that the optimum number of fog nodes does not increase in direct proportion to the increase in the total number of nodes. Furthermore, the optimum number of fog nodes decreases when channels have high path loss exponents. These findings suggest that the fog nodes must be selected among those that have the highest computation capability for densely deployed networks and high path loss exponents channels.
\end{abstract}


\section{Introduction}
A multitude of applications from augmented reality to online gaming and use cases from autonomous vehicles to smart cities in IoT/5G wireless networks are expected to produce an extraordinary increase in the amount of data. Although such a large-scale increase in data can be processed to some extent by cloud computing, the continuously growing amount of data cannot be tackled solely by cloud computing and fog computing has emerged as a promising method to accommodate the expected demands \cite{Bonomi}. Combining the large-scale data processing capability of cloud computing with the location aware, widely geographically distributed, low latency data processing capability of fog computing, is expected to be an attractive approach \cite{Bonomi}-\cite{YiFog}. This integration of cloud and fog is quite useful to process some portion of data in the network by fog computing, while processing the rest of data with cloud computing. The complementary nature of cloud and fog processing in data processing comprises a hierarchical network architecture dubbed a $\textit{cloud-fog-thing}$ network model \cite{Bonomi}-\cite{YiFog}. This architecture can be seen as a good compromise between fully centralized cloud networking and fully distributed fog networking.%, in which the former provides highly reliable data service with greater service delay, and the latter provides low latency service with less reliability. 

Maximizing the average signal-to-interference-plus-noise ratio (SINR) of the promising $\textit{cloud-fog-thing}$ network is of paramount importance for future applications and use cases, because this can enhance the average data rate, and thus improve the transmission delay that leads to a decrease in latency, which is a significant impact on the quality of user experience. In this regard, it is important to optimize the number of fog nodes to maximize the average SINR and data rate so as to minimize the transmission delay. Finding the optimum number of fog nodes, which will be upgraded from the existing nodes that have communication, computation and storage capability \cite{Bonomi},\cite{Chiang}, will further enhance the understanding and impact of the $\textit{cloud-fog-thing}$ architecture. To illustrate, there can be many potential nodes inside a network that will be upgraded to fog nodes. However, it is not clear why one does not update all the potential nodes to fog nodes to exploit all the available unused resources in the network, i.e., what is the incentive behind this? 

The idea of fog networking is clearly outlined with its benefits in \cite{Bonomi},\cite{Chiang}. Furthermore, the hierarchical $\textit{cloud-fog-thing}$ network is justified with different examples in \cite{Tang}-\cite{YiFog}. In addition, fog computing based radio access networks (RAN) are discussed in \cite{Peng},\cite{Hung}. One of the primary ideas common to all these papers is to upgrade some number of nodes into a fog node. However, the optimum number of nodes that will be upgraded to fog nodes as well as the incentive of not upgrading all nodes to a fog node are not stated. This study aims to fulfill this gap in the literature of the $\textit{cloud-fog-thing}$ network architecture. 

A stochastic geometry analysis is used to determine the optimum number of fog nodes that will be upgraded from the given number of ordinary nodes within the area of interest. It is crucial to emphasize that the widely used Poisson Point Process (PPP) model in stochastic geometry is not applicable for this case for two reasons. First, the PPP gives accurate models only for large-scale networks \cite{Banani} whereas a local fog network, e.g., residing in a park covers a local area, which constitutes a low-to-medium scale network. Second, more importantly, the total number of nodes is known in our case, which turns the Poisson process into a Binomial process. Note that the aforementioned hierarchical network topology is infrastructure based, and hence the total numbers of nodes is known. As a result, the Binomial Point Process (BPP) better represents the low-to-medium scale network whose total number of nodes is known \cite{Srinivasa}. 

The optimum number of fog nodes will be found by assuming that each node elects itself as a fog node with some probability, $p$. Then, the number of fog nodes becomes $np$, if there are $n$ number of nodes within the area of interest. The same approach is used to determine the cluster-heads or leaders of each cluster in wireless sensor networks \cite{Heinzelman}-\cite{Youssef2}, however, all those papers assume that the probability of being a cluster-head is given as $\textit{a priori}$ information instead of determining this by analysis. \cite{Bandyopadhyay} determines the optimum cluster-head probability using a PPP model in terms of energy efficiency for wireless sensor networks, which has some different notions than fog networking and is quite different than the situation addressed in this paper where the probability of being a fog node is found using BPP model. 

As stated in a recent survey paper, determining the optimum number of fog nodes is uncharted though finding it is quite important, and affect the overall network efficiency \cite{YMao}. Based on this motivation, the optimum number of fog nodes is determined for channels with different path loss exponents using a BPP model. Interestingly, our analysis indicates that too large or too small number of fog nodes decreases the performance. In addition, how the fog nodes scale with the incremental total number of nodes for different channels is quantified. As an additional benefit, the optimum number of nodes that can be controlled by a fixed number of fog nodes will also be found in this paper. This analysis might be useful in the design of the efficient virtual machines in the cloud, finding the value of $K$ in $K$-means clustering algorithm, which may be used to find the optimum locations of fog nodes, as well as enhancing the caching efficiency.

The paper is organized as follows. The network model and the problem statement are presented in Section \ref{Problem Statement}. In Section \ref{Problem Formulation}, the problem is formulated to find the optimum number of fog nodes. Section \ref{Analysis} introduces a stochastic geometry analysis for a BPP model. The derived closed-form derivations are validated in Section \ref{Validation} and the benefits and planned future research are given in Section \ref{Future}. The paper ends with the concluding remarks in Section \ref{Conclusions}.

\section{Network Model and Problem Statement}\label{Problem Statement}
Various applications in 5G and beyond wireless networks require an interplay between cloud and fog networks. Accordingly, some portion of data is processed at the fog networks and the remaining portion of data, i.e., filtered data is conveyed to the cloud. In this model, the inherent features of fog layer such as widely deployed geographical distribution and location awareness is associated with the large-scale data management capability of the cloud layer. The system of smart traffic light is one example that illustrates the interplay between cloud and fog networks so that distributed traffic lights connected to each other as well as vehicles, pedestrians and bikes intelligently control the traffic \cite{Bonomi}. Another example is a smart pipeline monitoring system in which the combination of fog and cloud networks sequentially process the data coming from the massive number of sensors \cite{Tang}. The same network model is highlighted in \cite{Chiang}, \cite{Luan} as well. In these papers, the network model is composed of the hierarchical combination of fog and cloud termed as $\textit{cloud-fog-thing}$ network as depicted in Fig. \ref{fig:NetworkModel}. Here, the fog layer is composed of many local fog networks located at parks, shopping malls, restaurants to name a few, where fog networks are made up of fog nodes that have communication, computation and storage capability, and emerged by updating the existing nodes in the network \cite{Bonomi}, \cite{Chiang}. The thing layer involves the end devices that may be various type of sensors, IoT devices or mobile phones.  
\begin{figure} [!h] 
\centering 
\includegraphics [width=3.5in]{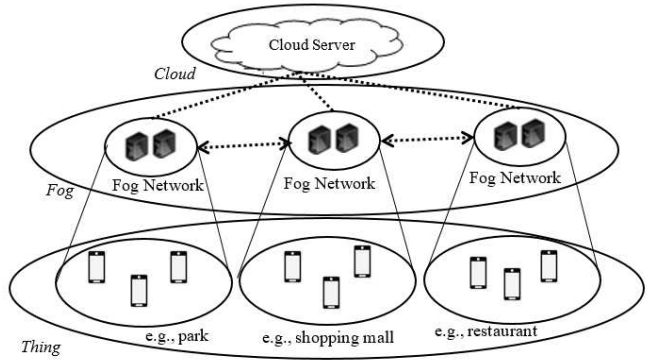}
\caption{$\textit{Cloud-fog-thing}$ type hierarchical network model}\label{fig:NetworkModel}
\end{figure}

Bearing in mind the overall network structure and the general notion of fog networking so that any existing node in the network can be a fog node, which is accepted by both industry \cite{Bonomi} and academy \cite{Chiang}, raise some fundamental questions. In this sense, this paper addresses the optimum number of fog nodes that should be upgraded from the existing nodes in the network. Assume that there is a square planar whose one side is $2a$, i.e., from -$a$ to $a$ and the cloud is located at the center, and the $n$ number of nodes are randomly and uniformly distributed around the cloud. At the beginning, these nodes are assumed to be ordinary, and then, some of them are specialized as fog nodes that constitute the fog layer and some of them remains ordinary that constitute the thing layer. To find the number of fog nodes, our approach is as follows. Each node selects itself as a fog node with probability $p$, and this yields $n_0$ and $n_1$ number of ordinary nodes and fog nodes respectively as $n_0=n(1-p)$, and $n_1=np$.

Needless to say that there has to be a criterion to determine the optimum probability of being a fog node $p$, and thus the optimum number of fog nodes. In this analysis, the criterion to find the optimum number of fog nodes is to maximize the average SINR, and thus data rate so as to minimize the transmission delay. Hence, the probability of being a fog node $p$ is optimized, and the optimum values of $n_0$ and $n_1$ will be found accordingly. In general, this paper provides a mathematical framework to specify the optimum number of fog nodes under one fog network so that one can find the optimum number of fog nodes dynamically even if the total number of nodes changes. Depending on this framework, one can determine the maximum possible nodes that can be controlled by a fixed number of fog nodes as well.

\section{Problem Formulation}\label{Problem Formulation}
A stochastic geometry analysis is performed to be able to formulate the optimum number of fog nodes when the end devices send their packets to the fog nodes, which forward the data to the cloud after processing some part of it. Accordingly, fog nodes and end devices are considered as points in $2$-dimensional Euclidean space. Throughout our analysis, it is assumed that the total number of points residing in the area of interest is known, though the number may dynamically change. Additionally, the number of nodes in the fog layer and in the thing layer may change. By this is meant that some nodes in the fog layer may be downgraded to the nodes in the thing layer or vice versa depending on the change in the network geometry due to mobility, or arrival or departure of nodes in the network. A widely used PPP model to accurately model the large-scale networks for random number of nodes in stochastic geometry \cite{Haenggi} cannot be applied to this problem, because the total number of nodes is known. Indeed, this knowledge turns a PPP into a BPP model \cite{Srinivasa}. Furthermore, the sub-regions covered by fog networks are not large-scale, i.e., they may be classified as low-to-medium scale network. Relying on these factors, it is more appropriate to model the underlying network model as a BPP.

The $\textit{cloud-fog-thing}$ network architecture can be simplified as a hierarchical tree based topology for one fog network as depicted in Fig. \ref{fig:TreeModel}. Here, nodes in the thing layer are termed as end devices that constitute Tier-0, which are controlled by the fog nodes located at Tier-1 and the cloud server is situated at the top layer. Note that fog nodes are connected to each other in a circular, fully connected mesh topology, and form the fog network, which is a generic and an appropriate model consistent with the definition of a fog network \cite{Bonomi}, \cite{Chiang}. There is an interplay between the number of nodes at Tier-1 and Tier-0 so that the number of fog nodes will be dynamically determined according to the number of end devices. More specifically, suppose that there are $n_0=n(1-p)$ and $n_1=np$ number of end devices and fog nodes, respectively, and $n = n_0+n_1$. To find the relation among $n$, $n_0$, and $n_1$, the optimum probability of being fog node $p$ is found. 
\begin{figure} [!h] 
\centering 
\includegraphics [width=3.5in]{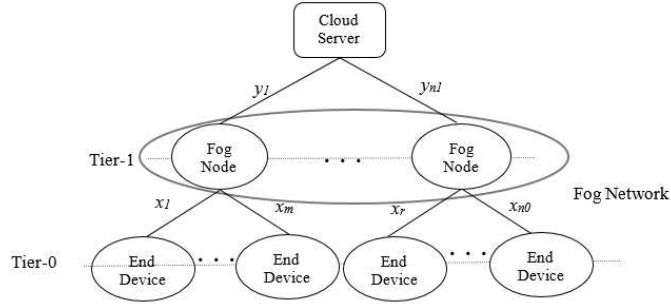}
\caption{A simplified tree based hierarchical network model}\label{fig:TreeModel}
\end{figure}

Assume that the packet size is $M$ bits and the packet is partially processed in the fog node, e.g., $K$ bits of the packet are processed, and the rest, i.e., $M-K$ bits are relayed to the cloud. Then, the one-way latency, $l$, from the end device to the cloud for this $2-$hop transmission, which is the end device-fog node-cloud, becomes
\begin{equation}\label{1}
l=\tau_{trans}+\tau_{proc} +\tau_{queue}+\tau_{prop}
\end{equation}
where $\tau_{trans}$ is the transmission delay, i.e.,
\begin{equation}\label{1.5}
\tau_{trans}=\frac{M}{R_{fog}}+\frac{M-K}{R_{cloud}}
\end{equation}
and $\tau_{proc}$ is the processing delay, $\tau_{queue}$ is the queueing delay, $\tau_{prop}$ is the propagation delay. $R_{fog}$ and $R_{cloud}$ are the data rate at the fog node and cloud as
\begin{equation}\nonumber
R_{fog}=Wlog(1+SINR_{fog})
\end{equation}
and
\begin{equation}\nonumber
R_{cloud}=Wlog(1+SINR_{cloud})
\end{equation}
where $W$ is the bandwidth, $SINR_{fog}$ and $SINR_{cloud}$ are the SINR at the fog node and cloud, respectively. More specifically, the SINR of the $i^{th}$ end device for $i=1,2,\cdots,n_0$ at the fog node becomes
\begin{equation}\label{2}
SINR_{fog}(i)=\frac{P_ih_ix_i^{-\alpha}}{\sigma^2+I_{fog}}
\end{equation}
where $P_i$ is the transmission power of the $i^{th}$ end device, $h_i$ is the channel power coefficient, $x_i$ is the distance between the end device and fog node as shown in Fig. \ref{fig:TreeModel}, $\alpha$ is the path loss coefficient, $\sigma^2$ is the noise variance and $I_{fog}$ is the residual interference power at the fog node after some interference mitigation techniques whose detailed discussion is out of scope for this paper. Notice that $I_{fog}=0$ in the idealized case, i.e., if the interference is perfectly mitigated. Similarly, the SINR due to the $j^{th}$ fog node for $j=1,2,\cdots,n_1$ at the cloud can be written as
\begin{equation}\label{3}
SINR_{cloud}(j)=\frac{P_jh_jy_j^{-\alpha}}{\sigma^2+I_{cloud}}
\end{equation}
where $P_j$ is the transmission power of the $j^{th}$ fog node, $h_j$ is the channel power coefficient, $y_j$ represents the distance between the $j^{th}$ fog node and the cloud, which is depicted in Fig. \ref{fig:TreeModel} as well. $I_{cloud}$ is the residual interference power at the cloud. Similarly, if one makes the assumption of perfect interference mitigation, it becomes $0$.

\begin{figure*}[!h] 
\begin{equation}\label{4}
J_{ij}=\min\left(\frac{1}{log(1+SINR_{fog}(i))}+\frac{1}{log(1+SINR_{cloud}(j))}\right)
\end{equation}
\end{figure*}
Consider the simple network structure that demonstrates the nodes given in Fig. \ref{fig:proof}. Here, circles represent the ordinary nodes, some of which will be upgraded to the fog nodes and the square denotes the cloud. Let's say that the distance between the circle that will not be upgraded as a fog node and be upgraded as a fog node is $\{x_i\}$ and the distance between a circle, i.e., the circle that will be upgraded to a fog node which is not known as \textit{a priori} and found after the optimization, and the square is $\{y_j\}$. In our formulation, fog nodes are selected based on $\{x_i\}$ and $\{y_j\}$, and thus the number of fog nodes are optimized accordingly. Notice that the selections of $\{x_i\}$ and $\{y_j\}$ are not affected by the processing and queueing delay, because all nodes are identical.
\begin{figure} [!h] 
\centering 
\includegraphics [width=3in]{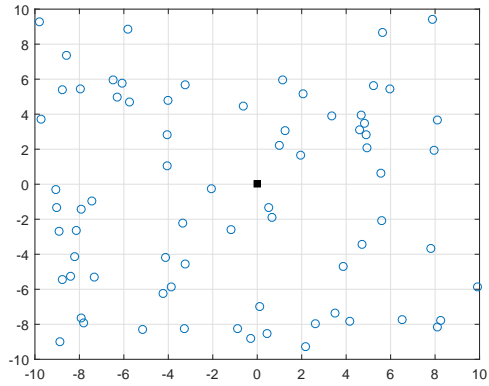}
\caption{The sample distribution of nodes for $a=10$}\label{fig:proof}
\end{figure}

Bearing in mind this network structure as well as the identical nodes, the objective function can be written as (\ref{4}) in terms of SINR that can maximize the data rate or minimize the $\tau_{trans}$ for the $i^{th}$ end device that transmits to the $j^{th}$ fog node, because $W$, $M$, $M-K$ are constant values. It is important to emphasize that $\tau_{proc}$ and $\tau_{queue}$ do not affect this optimization problem, because all nodes are identical. For the sake of simplicity, $\tau_{prop}$ is omitted as well, since the optimization that minimizes the right-hand side of (\ref{4}) with respect to distances can automatically minimize $\tau_{prop}$. Notice that the aim of this optimization is to find the fog numbers in a given area in terms of SINR instead of perfectly quantifying the latency. Hence, these simplifications are quite fair. 

Since the logarithm is a monotone function, (\ref{4}) is equivalent to
\begin{equation}\label{5}
\hat{J}_{ij}=\min\left(\frac{\sigma^2+I_{fog}}{P_ih_ix_i^{-\alpha}}+\frac{\sigma^2+I_{cloud}}{P_jh_jy_j^{-\alpha}}\right).
\end{equation}
Taking the expected value of (\ref{5}) produces
\begin{equation}\label{6}
\hat{J}_{ij}(avg)=\min\left(\frac{\sigma^2+I_{fog}}{P_i}E\left[\frac{x_i^{\alpha}}{h_i}\right]+\frac{\sigma^2+I_{cloud}}{P_j}E\left[\frac{y_j^{\alpha}}{h_j}\right]\right)
\end{equation}
for given $P_i$, $P_j$, $\sigma^2$, $I_{fog}$ and $I_{cloud}$, which may be either given as $\textit{a priori}$ information or estimated at the receiver, and thus they do not impress the optimization. Also, channel power coefficients are independent from distances that lead to
\begin{equation}\label{7}
\tilde{J}_{ij}(avg)=\min\left(E[x_i^{\alpha}]E\left[\frac{1}{h_i}\right]+E[y_j^{\alpha}]E\left[\frac{1}{h_j}\right]\right)
\end{equation}
where $E[1/h_i]=c_i$ and $E[1/h_j]=c_j$ such that $c_i$ and $c_j$ are constant values. This yields
\begin{equation}\label{8}
J_\alpha^{single}=\min\left(E[x_i^{\alpha}]+E[y_j^{\alpha}]\right).
\end{equation}

Solving (\ref{8}) gives one end device for one fog node that maximizes average SINR while a packet is sent from an end device to the cloud through a fog node. Since there are $n_0$ number of end devices and $n_1$ number of fog nodes, the objective function is defined as
\begin{equation} \label{11}
J_\alpha=\min \left(\sum_{i=1}^{n_0}E[x_i^{\alpha}]+\sum_{j=1}^{n_1}E[y_j^{\alpha}]\right)
\end{equation}
assuming that packets coming from the end devices to the fog nodes are aggregated, partially processed and relayed to the cloud. Since $n_0=n(1-p)$ and $n_1=np$, (\ref{11}) is optimized with respect to $p$, i.e., the value of $p$ that minimizes (\ref{11}) gives the number of fog nodes that will be upgraded from the ordinary nodes, which are randomly spatially distributed within the area of interest.

\section{The Optimum Number of Fog Nodes}\label{Analysis}
The number of fog nodes for each fog network can be optimized with respect to the objective function in (\ref{11}). In the model, it is assumed that there are $n$ numbers of nodes within the area of interest including $n_0$ numbers of end devices and $n_1$ numbers of fog nodes so that $n=n_0+n_1$. To find the interplay among the $n_0$ numbers of end devices and $n_1$ numbers of fog nodes, assume that the probability of being a fog node is $p$ for all $n$ nodes. This produces $n_0=n(1-p)$, $n_1=np$ numbers of end devices and fog nodes, respectively. Here, the critical point is the determination of $p$. Accordingly, first the objective function stated in (\ref{11}) will be derived as a closed-form expression in terms of $p$. Next, the objective function is optimized with respect to $p$ which determines the optimum values of $n_0$ and $n_1$. Notice that $p=0$ refers to the fact that there is no fog node whereas $p=1$ shows that all nodes must be fog nodes within the area of interest.

In particular, the values of $n_0$ and $n_1$ are optimized for $\alpha=1$, $\alpha=2$ and $\alpha=4$ in this paper. Within this aim, first a hypothetical condition is studied for $\alpha=1$. Although one can use the analysis of $\alpha=1$ as an approximation when the nodes are connected to a cable, this is physically meaningless for wireless connections. The main reason for analyzing the case for $\alpha=1$ is to better specify the relation between the optimum number of fog nodes and the path loss coefficient. Following that, the analysis is given for a free space path loss, i.e., $\alpha=2$. Lastly, a more practical case is considered for $\alpha=4$ accounting for the impact of shadowing and fading. 
\subsection{Hypothetical Path Loss}
The objective function in (\ref{11}) is first obtained as a closed-form expression in case of $\alpha=1$, which is a physically meaningless, but a mathematically meaningful quantity in wireless channels. This yields the following objective function 
\begin{equation} \label{12}
J_1= \min (x+y)
\end{equation}
where
\begin{equation} \label{13}
y=\sum_{j=1}^{n_1}E[y_j]
\end{equation}
and
\begin{equation} \label{14}
x=\sum_{i=1}^{n_0}E[x_i].
\end{equation}

All nodes are independently and uniformly distributed in a given square area of side $2a$ for $2$-dimensional Euclidean space with coordinates $(i_x, i_y)$. Depending on that, the expected distance of a fog node from the cloud can be expressed as 
\begin{equation} \label{15}
E[y_j]=\frac{1}{4a^2}\int_{-a}^{a}\int_{-a}^{a}\sqrt{i_x^2+i_y^2}d_{i_x}d_{i_y}=0.765a.
\end{equation}
Based on (\ref{15}), the total distance between the fog nodes and the cloud given in (\ref{13}) can be written for $n_1=np$ numbers of fog nodes as
\begin{equation} \label{16}
y=\sum_{j=1}^{n_1}E[y_j]=0.765npa.
\end{equation}
Following that, the average distance between two arbitrarily located points in a BPP is required to find the (\ref{14}). Specifically, the mean distance between a fog node and an end device is needed. At this point, a recently derived formula is used, which specifies the mean distance between two points for an isotropic BPP \cite{Srinivasa}, as  
\begin{equation} \label{17}
E[x_i]=\frac{ri^{1/2}}{(N+1)^{1/2}}
\end{equation}
where $r$ is the maximum range of the fog node, and $N$ is the total number of end devices controlled by a fog node, which becomes $N=(n-np)/np$ for the sake of analytical brevity. This produces
\begin{equation} \label{18}
E[x_i]=\frac{ri^{1/2}}{((n-np)/np+1)^{1/2}}.
\end{equation}

The maximum range $r$ can be specified by considering that each fog node, which is located at a center of a $2$-dimensional ball $b(o,r)$, has identical range and constitutes non-overlapping partitions without any loss of generality. This leads to
\begin{equation} \label{19}
r=\frac{\pi R}{np}
\end{equation}
where $R$ denotes the radius of the circular mesh fog network. Based on these, the sum distance between a fog node and end devices can be written as
\begin{equation} \label{20}
\sum_{i=1}^{n_0/n_1}E[x_i]=\sum_{i=1}^{n_0/n_1}\frac{\pi Ri^{1/2}}{(np)( (n-np) / np +1)^{1/2}}
\end{equation}
and the total distances due to having $np$ fog nodes can be given by
\begin{equation} \label{21}
x=\sum_{i=1}^{n_0/n_1}\frac{\pi Ri^{1/2}}{( (n-np) / np +1)^{1/2}}.
\end{equation}
After some mathematical operations, (\ref{21}) can be simplified as
\begin{equation} \label{22}
x=\frac{\pi R(n-np)}{np}.
\end{equation}
The objective function $J_1$ in (\ref{12}) can be given depending on (\ref{16}) and (\ref{22}) as
\begin{equation} \label{23}
J_1 = \frac{\pi R(n-np)}{np}+0.765npa. 
\end{equation}

\begin{lemma} \label{lemma1}
There is a unique optimum global value of $p$ that minimizes (\ref{23}).
\end{lemma}
\begin{proof}
The first and second derivative of (\ref{23}) with respect to $p$ can be simply written after some straightforward mathematical calculations as
\begin{equation} \label{24}
\frac{\partial J_1}{\partial p}=\frac{153anp^2-200\pi R}{200p^2}
\end{equation}
and
\begin{equation} \label{25}
\frac{\partial^2 J_1}{\partial p^2}=\frac{2\pi R}{p^3}.
\end{equation}
Since the second derivative of (\ref{25}) is greater than $0$, this means that (\ref{23}) is strictly convex function and the value of $p$ that makes (\ref{24}) $0$ is a global minimum point and unique, if it exists. Then, the question is whether a real $p$ exists or not. After some straightforward calculations, $p$ can be found as
\begin{equation} \label{26}
p=\left(\frac{200\pi R}{153an}\right)^{1/2}
\end{equation}
and hence it is a global minimum point and unique. 
\end{proof}

Due to Lemma \ref{lemma1}, the optimum number of fog nodes for a given total $n$ number of nodes can be calculated as
\begin{equation} \nonumber
n_1=\left(\frac{200\pi Rn}{153a}\right)^{1/2}.
\end{equation}
Alternatively, one can determine the optimum number of end devices in two steps if the number of fog nodes $n_1$ is known. Accordingly, first the value of $p$ is found as
\begin{equation} \nonumber
p=\frac{200\pi R}{153an_1}.
\end{equation}
Second, the number of end devices is calculated as 
\begin{equation} \nonumber
n_0=\frac{n_1}{p}-n_1.
\end{equation}
\subsection{Free Space Path Loss}
The signal power falls off with path loss exponents of $\alpha>1$ for wireless channels depending on many impediments. Considering only the effect of free space path loss causes $\alpha=2$. To find the optimum number of fog nodes for the path loss exponent of $2$, the objective function can be specified as
\begin{equation} \label{27}
J_2= \min (\tilde{x}+\tilde{y})
\end{equation}
where
\begin{equation} \label{28}
\tilde{y}=\sum_{j=1}^{n_1}E[y_j^2]
\end{equation}
and
\begin{equation} \label{29}
\tilde{x}=\sum_{i=1}^{n_0}E[x_i^2].
\end{equation}

The closed-form derivation of (\ref{28}) can be obtained as 
\begin{equation} \label{30}
E[y_j^2]=\frac{1}{4a^2}\int_{-a}^{a}\int_{-a}^{a}(i_x^2+i_y^2)d_{i_x}d_{i_y}=2a^2/3.
\end{equation}
Generalizing (\ref{30}) for all $np$ number of fog nodes produces
\begin{equation} \label{31}
\tilde{y}=\sum_{j=1}^{n_0}E[y_j^2]=2npa^2/3.
\end{equation}
On the other hand, the second moment of a distance between two arbitrarily located nodes in a BPP is derived in \cite{Srinivasa} as
\begin{equation} \label{32}
E[x_i^2]=\frac{r^2i}{(N+1)}.
\end{equation}
Integration of (\ref{32}) into our formulation gives
\begin{equation} \label{33}
\sum_{i=1}^{n_0/n_1}E[x_i^2]=\sum_{i=1}^{n_0/n_1}\frac{\pi^2R^2i}{(np)^2((n-np) / np+1)}=\frac{(n-np)\pi^2R^2}{2(np)^3}.
\end{equation}
Generalizing (\ref{33}) for all $np$ number of fog nodes results in
\begin{equation} \label{34}
\tilde{x}=\frac{(n-np)\pi^2R^2}{2(np)^2}.
\end{equation}
Depending on (\ref{31}) and (\ref{34}), the objective function in (\ref{27}) ends up as
\begin{equation} \label{35}
J_2 = \frac{(n-np)\pi^2R^2}{2(np)^2}+\frac{2npa^2}{3} 
\end{equation}

\begin{lemma} \label{lemma2}
There is a unique optimum global value of $p$ that minimizes (\ref{35}).
\end{lemma}

\begin{proof} 
One can easily show that (\ref{35}) is a concave upward function by inspecting its second derivative. For this purpose, the first and second derivative of (\ref{35}) is consecutively written as
\begin{equation} \label{36}
\frac{\partial J_2}{\partial p}=\frac{4a^2n^2p^3+3\pi^2R^2p-6\pi^2R^2}{6np^3}
\end{equation}
and
\begin{equation} \label{37}
\frac{\partial^2 J_2}{\partial p^2}=\frac{\pi^2R^2(3-p)}{np^4}.
\end{equation}
It is clear from (\ref{37}) that the acquired objective function in (\ref{35}) is strictly convex for $0<p<1$. This means that any real root of (\ref{36}) minimizes the objective function if such kind of a $p$ exists. After some mathematical operations, it can be proved that
\begin{equation} \label{38}
p=\left(\frac{6\pi^2R^2}{4a^2n^2}\right)^{1/3}.
\end{equation}
Hence, (\ref{38}) is the unique global minimum point, which minimizes (\ref{35}). 
\end{proof}

Suppose that there are $n$ number of nodes in an area. According to Lemma \ref{lemma2}, the optimum fog number for this area becomes
\begin{equation} \nonumber
n_1=\left(\frac{6\pi^2R^2n}{4a^2}\right)^{1/3}.
\end{equation}
Analogously, if the number of fog nodes is known as $\textit{a priori}$ information, then one can easily find $n_0$ as
\begin{equation} \nonumber
n_0=\frac{n_1}{p}-n_1
\end{equation}
where
\begin{equation} \nonumber
p=\frac{6\pi^2R^2}{4a^2n_1^2}.
\end{equation}

\subsection{Shadowing and Fading}
In wireless channels, shadowing and fading are other factors that affect the path loss exponent in addition to free space path loss. Accounting for the impact of all free space path loss, shadowing and fading, it is reasonable to take $\alpha=4$ \cite{Goldsmith} and formulate the objective function accordingly. More rigorously, 
\begin{equation} \label{39}
J_4= \min (\hat{x}+\hat{y})
\end{equation}
where 
\begin{equation} \label{40}
\hat{y}=\sum_{j=1}^{n_1}E[y_j^4]
\end{equation}
and
\begin{equation} \label{41}
\hat{x}=\sum_{i=1}^{n_0}E[x_i^4].
\end{equation}

The expression in (\ref{40}) can be written similar to (\ref{15}) and (\ref{30}) as 
\begin{equation} \label{42}
E[y_j^4]=\frac{1}{4a^2}\int_{-a}^{a}\int_{-a}^{a}(i_x^2+i_y^2)^2d_{i_x}d_{i_y}=0.62a^4
\end{equation}
which results in
\begin{equation} \label{43}
\hat{y}=\sum_{j=1}^{n_0}E[y_j^2]=0.62npa^4.
\end{equation}
The fourth moment of a distance that belongs to two points in a BPP can be specified as \cite{Srinivasa}
\begin{equation} \label{44}
E[x_i^4]=\frac{r^4i^2}{(N+1)^2}.
\end{equation}
Regarding our problem, (\ref{44}) can be interpreted as 
\begin{equation} \label{45}
\sum_{i=1}^{n_0/n_1}E[x_i^4]=\sum_{i=1}^{n_0/n_1}\frac{\pi^4R^4i^2}{(np)^4((n-np) / np+1)^2}.
\end{equation}
This leads to
\begin{equation} \label{46}
\hat{x}=\sum_{i=1}^{n_0/n_1}\frac{\pi^4R^4i^2}{(np)^3((n-np) / np+1)^2}\approx\frac{\pi^4R^4n(1-p)}{(np)^4}.
\end{equation}
Due to (\ref{43}) and (\ref{46}), $J_4$ can be given by
\begin{equation} \label{47}
J_4 = \frac{\pi^4R^4n(1-p)}{(np)^4}+0.62npa^4
\end{equation}

\begin{lemma} \label{lemma3}
There is a unique optimum global value of $p$ that minimizes (\ref{47}).
\end{lemma}

\begin{proof} 
The first and second derivative of (\ref{47}) becomes
\begin{equation} \label{48}
\frac{\partial J_4}{\partial p}=\frac{31a^4n^4p^5+150\pi^4R^4p-200\pi^4R^4}{50n^3p^5}
\end{equation}
and
\begin{equation} \label{49}
\frac{\partial^2 J_4}{\partial p^2}=\frac{4\pi^4R^4(5-3p)}{n^3p^6}
\end{equation}
respectively. It is worth emphasizing that (\ref{49}) is greater than $0$ for $0<p<1$ suggesting that the function in (\ref{47}) is strictly convex. Then, the real root of (\ref{48}), which is equal to
\begin{equation} \label{50}
p=\left(\frac{200\pi^4R^4}{31a^4n^4}\right)^{1/5}
\end{equation}
constitutes the global unique minimal point.
\end{proof} 

As a consequence, the optimum number of fog nodes is equal to 
\begin{equation} \nonumber
n_1=\left(\frac{200\pi^4R^4n}{31a^4}\right)^{1/5}.
\end{equation}
Let's assume that the number of fog nodes is given as $\textit{a priori}$ information, i.e., $n_1$ is known. Then, the optimum number of $n_0$ can be calculated as
\begin{equation}\nonumber
p=\frac{200\pi^4R^4}{31a^4n_1^4}
\end{equation}
so that
\begin{equation} \nonumber
n_0=\frac{n_1}{p}-n_1.
\end{equation}

\section{Validation of Analyses}\label{Validation}
To develop more insights about the optimum number of fog nodes in a given area for different path loss exponents, the derived closed-form objective functions in (\ref{23}), (\ref{35}), (\ref{47}) are numerically analyzed. In particular, these functions are plotted with respect to $p$ and the optimum values of $p$ that minimize the objective functions are found. The numerically derived $p$ values are compared with the values found analytically with some approximations in (\ref{26}), (\ref{38}) and (\ref{50}). Following that, the optimum number of fog nodes and the average number of end devices that are controlled by a single fog node are specified for a given $n$. Lastly, the improvement of the average SINR due to the optimized number of fog nodes is quantified in terms of data rate that can trivially affect the latency considering the transmission delay.

In our network model, a cloud covers a large area, and thus $a$ is selected as $50$ km, which corresponds to an area of $10000$ km$^2$. On the other hand, each local fog network is responsible for a relatively small area. The radius of the circular fog network is taken as $R=0.0765a$ without any loss of generality, which corresponds to one-tenth of the average distance between the fog node and cloud, that covers an area of roughly $45$ km$^2$. The analysis is repeated for a total number of $200$, $400$ and $800$ nodes for one fog network and $\alpha=1$, $\alpha=2$ and $\alpha=4$.

In the first case, the objective function for $\alpha=1$ in (\ref{23}) is obtained for $n=200$, $n=400$ and $n=800$ as depicted in Fig. \ref{fig:BPP_alpha1}. Notice that as proven in Lemma \ref{lemma1}, each curve has a unique minimum point, and this point specifies the optimum number of fog nodes. Note that the ratio of fog nodes to the total number of nodes decreases with increasing $n$. This suggests that more devices should be handled by fog nodes in ultra-densely deployed networks. To be more specific, Table 1 gives the optimum number of fog nodes and specifies the average end devices in each fog node for different values of $n$. Furthermore, the minimum value of $p$ found in (\ref{26}) for $\alpha=1$ is compared with the one that is obtained with numerically in Table \ref{tab:Table1}.
\begin{figure} [!h] 
\centering 
\includegraphics [width=3.5in]{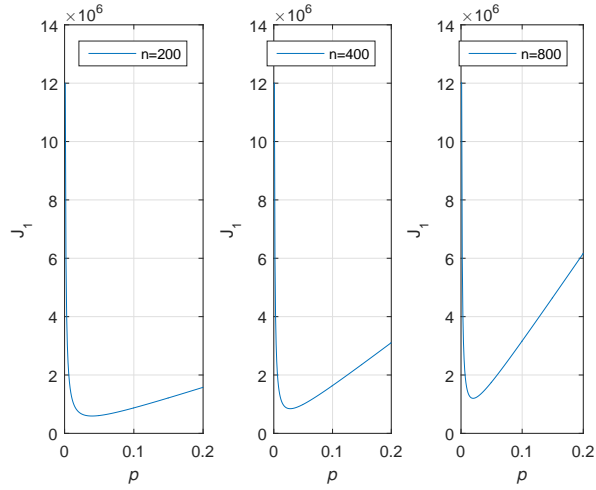}
\caption{The objective function for $\alpha=1$ in terms of $p$}\label{fig:BPP_alpha1}
\end{figure}

\begin{table}
\centering
\caption{The comparison of analytical and numerical results and the optimum number of fog nodes with average end devices for $\alpha=1$} \label{tab:Table1}
\begin{tabular}{ |p{1.25cm}||p{1.25cm}|p{1.25cm}| p{1.25cm}|p{1.25cm}|} 
\hline
 & Analytical value of $p$ in (\ref{26}) & Numerical value of $p$ & Average number of end devices & Average optimum number of fog nodes \\
\hline
$n = 200$ & $0.0396$ & $0.04$ & $24.25$ & $7.92$\\
$a = 400$ & $0.028$ & $0.028$ & $34.71$ & $11.2$\\
$a = 800$ & $0.0198$ & $0.02$ & $49.5$ & $15.84$\\
\hline
\end{tabular}
\end{table}

A simulation is performed to compare the average data rate for the optimized and unoptimized number of fog nodes based on the values in Table \ref{tab:Table1}. Accordingly, it is assumed that there are $200$ nodes within the area of interest without any loss of generality. For the former case, i.e., the optimized number of fog nodes, there are nearly $8$ nodes at the fog layer and $192$ nodes at the thing layer. For the latter unoptimized case, we randomly generate different number of fog nodes and take their average. Then, the ratio of the average data rate for the optimized number of fog nodes $R_{opt}$ and the average data rate for the unoptimized number of fog nodes $R_{unopt}$ becomes as illustrated in Fig. \ref{fig:ple1avgDataRate}. As one can observe, the ratio slightly decreases for higher signal-to-noise ratio (SNR) values, however, there is still significant advantage of optimizing the number of fog nodes. Note that this benefit will grow for high bandwidths. Additionally, it is straightforward to see this effect in latency considering transmission delay, which is obtained by dividing the packet size with data rate.
\begin{figure} [!h] 
\centering 
\includegraphics [width=3.5in]{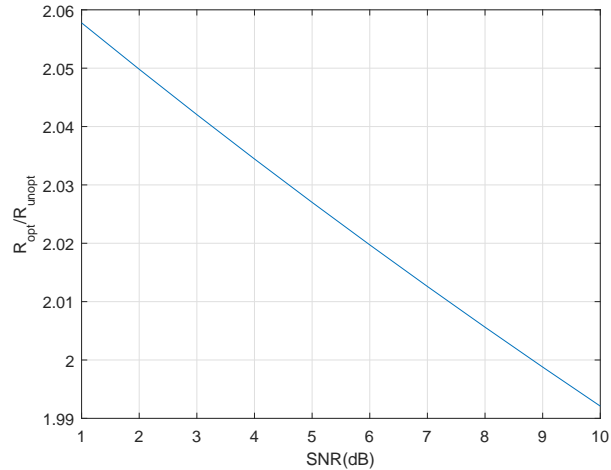}
\caption{The advantage of the optimized number of fog nodes in terms of the average data rate when $\alpha=1$.}\label{fig:ple1avgDataRate}
\end{figure}

The same experiment is repeated for a path loss exponent of $\alpha=2$ that represents the free space path loss, whose closed-form objective function is given by (\ref{35}). The results for $\alpha=2$ reveals similar characteristics as with $\alpha=1$ and is depicted in Fig. \ref{fig:BPP_alpha2}. Here, the numerically determined $p$ is nearly the same as the analytical result in (\ref{38}) as clarified in Table \ref{tab:Table2}, which also introduces the optimum number of fog nodes in one fog network and the average number of end devices in one fog node. One more important point is that the optimum number of fog nodes inside a fog network decreases with respect to the previous case, i.e., $\alpha=1$ for the same number of $n$. Intuitively, for nodes whose signal power falls more rapidly, the fog nodes that are further from the cloud with respect to the other fog nodes will significantly decrease the performance. This means that it may be more advantageous to send the packets to the closest fog node instead of being a fog node that decreases the overall optimum number of fog nodes. In this case, the average number of end devices controlled by a fog node increases. This result indicates that more computational power is necessary for the fog nodes in case of channels with high path loss exponents.
\begin{figure} [!h] 
\centering 
\includegraphics [width=3.5in]{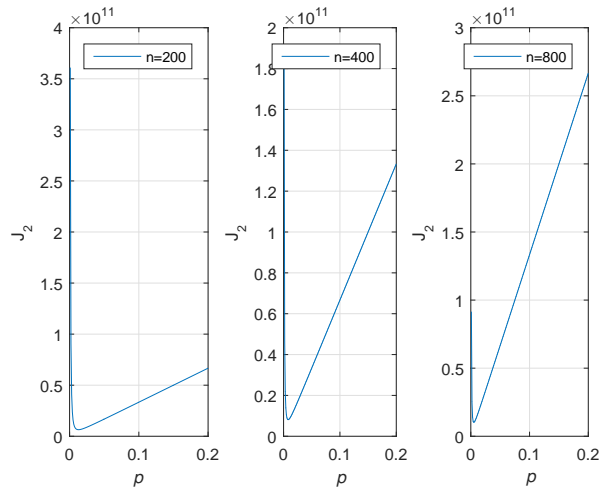}
\caption{The objective function for $\alpha=2$ in terms of $p$}\label{fig:BPP_alpha2}
\end{figure}

\begin{table}
\centering
\caption{The comparison of analytical and numerical results and the optimum number of fog nodes with average end devices for $\alpha=2$} \label{tab:Table2}
\begin{tabular}{ |p{1.25cm}||p{1.25cm}|p{1.25cm}| p{1.25cm}|p{1.25cm}|} 
\hline
 & Analytical value of $p$ in (\ref{38}) & Numerical value of $p$ & Average number of end devices & Average optimum number of fog nodes \\
\hline
$n = 200$ & $0.0129$ & $0.013$ & $76.51$ & $2.58$\\
$a = 400$ & $0.0082$ & $0.008$ & $120.95$ & $3.28$\\
$a = 800$ & $0.0051$ & $0.005$ & $195.07$ & $4.08$\\
\hline
\end{tabular}
\end{table}

The ratio of the average data rate between the optimized and unoptimized number of fog nodes are evaluated for $\alpha=2$ as well in Fig. \ref{fig:ple2avgDataRate} when the total number of nodes is selected as $200$. Although the ratio of $R_{opt}/R_{unopt}$ shows little decreases with incremental SNR, it is higher than the case of $\alpha=1$. In fact, the data rate approximately doubles once the number of fog nodes is optimized, which is quite important in future wireless networks considering the increasing user demands.
\begin{figure} [!h] 
\centering 
\includegraphics [width=3.5in]{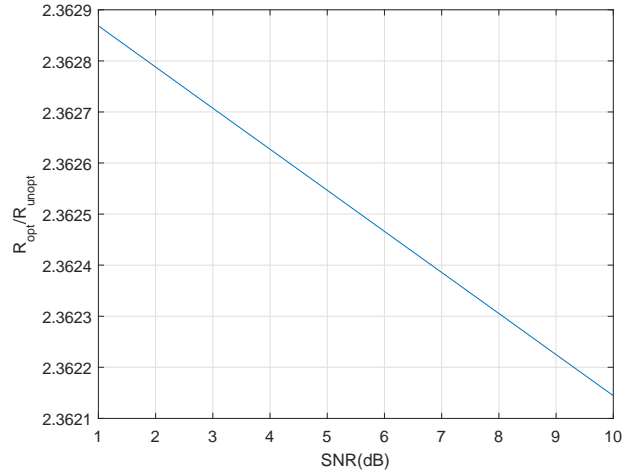}
\caption{The advantage of the optimized number of fog nodes in terms of the average data rate when $\alpha=2$.}\label{fig:ple2avgDataRate}
\end{figure}

Lastly, the channels that are subject to shadowing and fading are considered to determine the optimum number of fog nodes regarding the average data rate. Here, the path loss exponent is selected as $\alpha=4$ as in the derived objective function in (\ref{47}). The results for this case are illustrated in Fig. \ref{fig:BPP_alpha4} and Table \ref{tab:Table3}. Compared to $\alpha=1$ and $\alpha=2$, the fewer number of fog nodes, each of which has higher number of end devices, are needed for $\alpha=4$. It can be deduced that providing services to the end devices by a fog node is more challenging for channels with higher path loss exponents, because the increase in the number of end devices complicates the data processing and resource allocation.  
\begin{figure} [!h] 
\centering 
\includegraphics [width=3.5in]{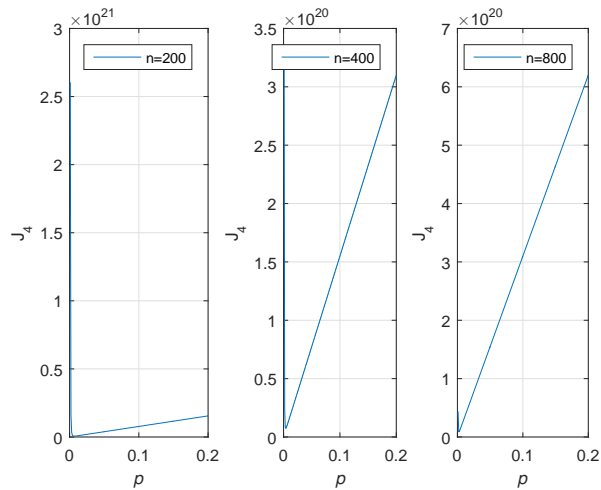}
\caption{The objective function for $\alpha=4$ in terms of $p$}\label{fig:BPP_alpha4}
\end{figure}

\begin{table}
\centering
\caption{The comparison of analytical and numerical results and the optimum number of fog nodes with average end devices for $\alpha=4$} \label{tab:Table3}
\begin{tabular}{ |p{1.25cm}||p{1.25cm}|p{1.25cm}| p{1.25cm}|p{1.25cm}|} 
\hline
 & Analytical value of $p$ in (\ref{50}) & Numerical value of $p$ & Average number of end devices & Average optimum number of fog nodes \\
\hline
$n = 200$ & $0.0067$ & $0.007$ & $148.25$ & $1.34$\\
$a = 400$ & $0.0038$ & $0.004$ & $262.15$ & $1.52$\\
$a = 800$ & $0.0022$ & $0.002$ & $453.54$ & $1.76$\\
\hline
\end{tabular}
\end{table}

When it comes to specifying the average data rate for the optimized number of fog nodes, there is a considerable enhancement with respect to the unoptimized one as depicted in Fig. \ref{fig:ple4avgDataRate}. More specifically, the average data rate increases almost be an order of magnitude if the number of fog nodes is optimized. Notice that the improvement is nearly the same for different SNR values. This emphasizes that it is much more important to optimize the number of nodes in the fog layer and the thing layer for channels with higher path loss exponents.
\begin{figure} [!h] 
\centering 
\includegraphics [width=3.5in]{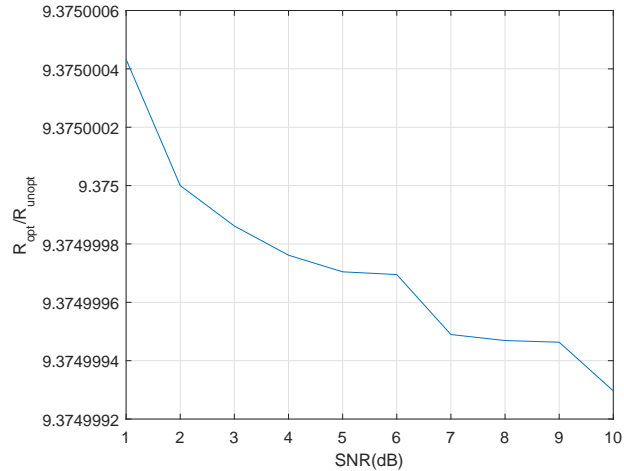}
\caption{The advantage of the optimized number of fog nodes in terms of the average data rate when $\alpha=4$.}\label{fig:ple4avgDataRate}
\end{figure}

The overall results imply that the boost in the number of total nodes leads to a moderate increase in the number of fog nodes, whereas there is a significant increase in the number of end devices. Furthermore, fewer fog nodes are sufficient in case of severe fading. This indicates that each fog node must have more resources to manage the needs of users when there are many nodes in channels with high path loss exponents. This fact indicates the importance of cooperation and virtualization in fog networks. An asymptotical situation can be to have a large number of path loss exponent channels that results in a single fog node with many end devices. As a result, the channels with smaller path loss exponent have a tendency to have more fog nodes implying more distributed networking while the channels with higher path loss exponents tend to more centralized solution to maximize the data rate. 

\section{Additional Benefits and Future Work}\label{Future}
Optimizing the number of fog nodes brings with it numerous benefits for the design of $\textit{cloud-fog-thing}$ networks. Regarding the cloud layer of this architecture, whose primary advantage comes from virtualization, the problem of underutilized or overutilized virtual machines is one of the problems that decreases the efficiency of virtualization. In general, to balance the loads in virtual machines that run on physical machines, prediction based algorithms that observe the past statistics are employed \cite{Xiao}. However, if the load dynamically changes, i.e., it is time-varying, these sorts of algorithms do not give accurate estimate. At this point, one may want to use the information of the dynamically optimized number of fog nodes derived in this paper while designing virtual machines, which can be viewed as fog-aided cloud virtualization. The basic idea here is to associate the virtual machines in the cloud with different fog nodes whose optimum number is determined according to the total number of nodes in the area of interest, and each node has a balanced average load. That is, if virtual machines were created such that each virtual machine became responsible for an equal number of fog nodes, all virtual machines would have more balanced loads. This could minimize the number of physical machines and lead to more efficient green computing as well. Notice that each fog node needs to communicate with a virtual machine in the cloud for many reasons, e.g., to update its cache, manage its resources more efficiently, or send some portion of data coming from end devices to the virtual machines for further processing. The details of this subject will be handled in future work.

Fog layer design requires not only knowing the optimum number of fog nodes but also finding the locations of fog nodes. That is, which nodes in the network are upgraded as fog nodes, among the many alternatives has to be determined. Clustering algorithms in machine learning can be used to address the locations of the fog nodes so that cluster-heads can become the fog nodes. One of the widely used clustering algorithm in machine learning is the $K$-means clustering algorithm based on the principle of minimizing inter-cluster distances \cite{Bishop}. Accordingly, the geographical locations of the potential fog nodes can be considered as a data set that can be clustered with $K$-means clustering algorithm so that the closest nodes to the leader of each cluster, i.e., cluster-heads, can be upgraded as fog nodes. Despite its simplicity and efficiency, the major drawback of this algorithm is in the determination of the value of $K$. It is not clear how one should select $K$, and there are only heuristics instead of mathematical analysis \cite{Jain},\cite{Tibshirani}. As a promising solution, the stochastic geometry analysis given in this paper can fulfill this gap. Specifically, the value of $K$, which is the optimum number of fog nodes, can be analytically obtained as $K=np$ where $n$ is given as $\textit{a priori}$ information and $p$ is derived as a closed-form expression in (\ref{26}), (\ref{38}) and (\ref{50}). The details of this subject will be explored in future work as well.

Quality of experience of users residing in the thing layer is highly related to the efficient caching mechanism in the fog layer. A recent paper reveals that the performance of caching depends on both the capacity of the front-haul network between the fog and the cloud, as well as the caching resources in the fog nodes for $\textit{cloud-fog-thing}$ network \cite{Tandon}. This means that even upgrading all of the nodes to fog nodes to exploit the unused resources for the sake of caching is not sufficient to have better caching performance. Therefore, the optimized number of fog nodes that can improve the average data rate within the network will enhance the front-haul capacity and affect the caching. It is worth noting that evaluating the caching performance quantitatively in terms of the number of fog nodes is a good research problem.

\section{Conclusions}\label{Conclusions}
A promising network model for 5G and beyond wireless applications that ensures the processing of big data based on the complementary nature of cloud and fog is studied in this paper. Specifically, the number of optimum fog nodes, which is one of the unclear points in the $\textit{cloud-fog-thing}$ architecture, is clarified in this paper using the tools of stochastic geometry. Our optimization enhances the average SINR in the network, and thus maximizes the average data rate and minimizes the transmission delay. In particular, it is quite meaningful and important to maximize the average data rate especially for the networks that require big data processing. Our results show that optimizing the number of fog nodes significantly improves the average data rate. To illustrate, the average data rate doubles and increases by almost an order of magnitude for the free-space path loss channels, and shadowing and fading channels, respectively. Indeed, having more than or less than the optimum number of fog nodes degrades the average data rate, and its effect becomes greater for the channels with high path loss exponents. Furthermore, the optimum number of fog nodes decreases for high path loss exponents channels indicating that fog nodes must be selected among the nodes that have the highest computational power for these channels. Consequently, this paper facilitates the integration of theoretical results on fog networking to practical networks. Within this scope, the optimum number of fog nodes is found. Our results may be quite useful in the design of cloud virtualization, while finding the locations of fog nodes with $K$-means clustering algorithm and having enhanced caching performance.

\end{document}